\documentclass[a4paper,UKenglish]{lipics2}

\usepackage[utf8]{inputenc}
\usepackage{enumitem}
\setlist[enumerate]{leftmargin=*}
\setlist[itemize]{leftmargin=*}

\theoremstyle{definition}
\newtheorem{observation}[theorem]{Observation}

\usepackage{cite}

\usepackage[disable]{todonotes}

\usepackage{microtype}

\title{CLS: New Problems and Completeness}
\titlerunning{CLS: New Problems and Completeness} 

\author[1]{John Fearnley}
\author[2]{Spencer Gordon}
\author[2]{Ruta Mehta}
\author[1]{Rahul Savani}
\affil[1]{University of Liverpool, Liverpool, UK\\
  \texttt{\{john.fearnley, rahul.savani\}@liverpool.ac.uk}}
\affil[2]{University of Illinois at Urbana-Champaign, Urbana IL, USA\\
  \texttt{rutamehta@cs.illinois.edu, slgordo2@illinois.edu}}
\authorrunning{J.\ Fearnley, S.\ Gordon, R.\ Mehta, and R.\ Savani} 

\Copyright{John Fearnley, Spencer Gordon, Ruta Mehta, and Rahul Savani}


\usepackage{tabularx,amssymb,nicefrac,amsmath,multirow,stmaryrd}
\usepackage{comment,amsfonts} 
\usepackage{epsfig} \usepackage{latexsym,nicefrac,bbm}
\usepackage{xspace}
\usepackage{color,fancybox,graphicx,url}
\usepackage{tabularx} 
\usepackage{booktabs}
\usepackage{color,colortbl}
\usepackage{mathtools}

\newcommand{\ra}{\rightarrow}
\long\def\symbolfootnote[#1]#2{\begingroup%
\def\thefootnote{\fnsymbol{footnote}}\footnote[#1]{#2}\endgroup}

\def\cc#1{\mathsf{#1}}
\def\CLS{\ensuremath{\cc{CLS}}\xspace}
\def\NP{\ensuremath{\cc{NP}}\xspace}

\def\TFNP{\ensuremath{\cc{TFNP}}\xspace}
\def\PPAD{\ensuremath{\cc{PPAD}}\xspace}
\def\PLS{\ensuremath{\cc{PLS}}\xspace}
\def\PPADPLS{\ensuremath{\cc{PPAD} \cap \cc{PLS}}\xspace}

\def\problem#1{\textsc{#1}}
\def\CM{\problem{Contraction}\xspace}
\def\GCM{\problem{GeneralContraction}\xspace}
\def\MMCM{\problem{MetametricContraction}\xspace}

\def\EOL{\problem{EndOfLine}\xspace}
\def\EOPL{\problem{EndOfPotentialLine}\xspace}
\def\EOML{\problem{EndOfMeteredLine}\xspace}
\def\CLO{\problem{ContinuousLocalOpt}\xspace}

\def\PLCP{\problem{P-LCP}\xspace}
\def\ContractionMap{\problem{ContractionMap}\xspace}
\def\MBanach{\problem{MetricBanach}\xspace}
\def\e{\epsilon}

\def\ite{\mbox{ItoE}}
\def\eti{\mbox{EtoI}}
\def\pot{\mbox{$V$}}
\def\isvalid{\mbox{IsValid}}
\def\PLo{\mbox{Q1}}
\def\PLt{\mbox{Q2}}

\def\Real{\mathbb{R}}

\def\Set#1{\left\{ #1 \right\}}
\def\Abs#1{\left| #1 \right|}

\def\Norm#1{\left\| #1 \right\|}

\makeatletter
\def\Bigbar#1{\mathrel{\left|\vphantom{#1}\right.\n@space}}

\def\vert{\operatorname{\mathsf{vert}}}

\def\begin@lgo{\begin{minipage}{1in}\begin{tabbing}
		\quad\=\qquad\=\qquad\=\qquad\=\qquad\=\qquad\=\qquad\=\qquad\=\qquad\=\qquad\=\qquad\=\qquad\=\qquad\=\kill}
\def\end@lgo{\end{tabbing}\end{minipage}}

\newenvironment{algorithm}
{\begin{tabular}{|l|}\hline\begin@lgo}
{\end@lgo\\\hline\end{tabular}}

\makeatother

\newcommand{\CPol}{\mbox{${\mathcal P}$}}
\newcommand{\CI}{\mbox{${\mathcal I}$}}
\newcommand{\CL}{\mbox{${\mathcal L}$}}
\newcommand{\CE}{\mbox{${\mathcal E}$}}
\newcommand{\CQ}{\mbox{${\mathcal Q}$}}

\newcommand{\yy}{\mbox{\boldmath $y$}}
\newcommand{\bb}{\mbox{\boldmath $b$}}
\newcommand{\uu}{\mbox{\boldmath $u$}}

\newcommand{\vv}{\mbox{\boldmath $v$}}
\newcommand{\qq}{\mbox{\boldmath $q$}}
\newcommand{\xx}{\mbox{\boldmath $x$}}
\newcommand{\one}{\mbox{\boldmath $1$}}
\newcommand{\ones}{\mbox{\boldmath $1$}}
\newcommand{\zeros}{\mbox{\boldmath $0$}}
\newcommand{\MM}{\mbox{$M$}}
\newcommand{\ps}{\mbox{\boldmath $s$}}
\newcommand{\pq}{\mbox{\boldmath $q$}}

\begin{document}

\maketitle



\begin{abstract}
The complexity class \CLS was introduced by Daskalakis and Papadimitriou in
	\cite{daskalakis2011continuous} with the goal of capturing the complexity of
	some well-known problems in $\PPAD \cap \PLS$ that have resisted, in some
	cases for decades, attempts to put them in polynomial time.  No complete
	problem was known for~\CLS, and in \cite{daskalakis2011continuous}, the problems
	\CM, i.e., the problem of finding an approximate fixpoint of a contraction map,
	and \PLCP, i.e., the problem of solving a P-matrix Linear Complementarity
	Problem, were identified as prime candidates. 

First, we present a new \CLS-complete problem \MMCM, which is closely related to
	the \CM. Second, we introduce \EOPL, which captures aspects of \PPAD and
	\PLS directly via a monotonic directed path, and show that \EOPL is in \CLS via
	a two-way reduction to \EOML. The latter was defined
	in~\cite{hubavcek2017hardness} to keep track of how far a vertex is on the
	\PPAD path via a restricted potential function.  Third, we reduce \PLCP to
	\EOPL, thus making \EOPL and \EOML at least as likely to be hard for \CLS as
	\PLCP. This last result leverages the monotonic structure of Lemke paths for
	\PLCP problems, making \EOPL a likely candidate to capture the exact
	complexity of \PLCP; we note that the structure of Lemke-Howson paths for
	finding a Nash equilibrium in a two-player game very directly motivated the
	definition of the complexity class \PPAD, which eventually ended up
	capturing this problem's complexity exactly. 

%
%
\end{abstract}

\section{Introduction}

The complexity class \TFNP, which stands for total function problems in
\NP, contains search problems that are guaranteed to have a solution, and whose
solutions can be verified in polynomial time~\cite{megiddo1991total}.
While \TFNP is a semantically defined complexity class and is thus unlikely to
contain complete problems, a number of syntactically defined subclasses of
\TFNP have proven very successful at capturing the complexity of total search 
problems.
For example, the complexity class $\cc{PPAD}$, introduced in
\cite{papadimitriou1994complexity} to capture the difficulty of search problems
that are guaranteed total by a parity argument, attracted intense attention in
the past decade culminating in a series of papers showing that the problem of
computing a Nash-equilibrium in two-player games is $\cc{PPAD}$-complete
\cite{chen2009settling,daskalakis2009complexity}. There are no known
polynomial-time algorithms for $\cc{PPAD}$-complete problems, and recent work
suggests that no such algorithms are likely to exist~\cite{bitansky2015cryptographic,garg2016revisiting}. 
The class of problems that
can be solved by local search (in perhaps exponentially-many steps), $\cc{PLS}$,
has also attracted much interest since it was introduced in
\cite{johnson1988easy}, and looks similarly unlikely to have polynomial-time
algorithms. Examples of problems that are complete for \PLS include the problem
of computing a pure Nash equilibrium in a congestion
game~\cite{fabrikant2004complexity} and computing a locally optimal max
cut~\cite{schaffer1991simple}.

If a problem lies in both \PPAD and \PLS then it is unlikely to be complete for 
either class, since this would imply a extremely surprising containment of one class in the other.
Motivated by the existence of several total function problems in \PPADPLS
that have resisted researchers attempts to design polynomial-time algorithms,
in their 2011 paper \cite{daskalakis2011continuous}, Daskalakis and Papadimitriou introduced
the class $\CLS$, a syntactically defined subclass of $\cc{PPAD} \cap \cc{PLS}$.
\CLS is intended to capture the class of optimization problems over a continuous
domain in which a continuous potential function is being minimized and the
optimization algorithm has access to an continuous improvement function.
Daskalakis and Papadimitriou showed that many classical problems of unknown
complexity were shown to be in $\CLS$ including the problem of solving a simple
stochastic game, the more general problem of solving a Linear Complementarity
Problem with a P-matrix, and the problem of finding an approximate fixpoint
to a contraction map. Moreover, $\CLS$ is the smallest known subclass of
$\cc{TFNP}$ and hardness results for it imply hardness results for
$\cc{PPAD}$ and $\cc{PLS}$ simultaneously. 


Recent work by Hub\'{a}\v{c}ek and Yogev~\cite{hubavcek2017hardness} proved 
lower bounds for \CLS. They introduced
a problem known as $\EOML$ which they showed was in $\CLS$, and for which they
proved a query complexity lower bound of $\Omega(2^{n/2}/\sqrt{n})$ and
hardness under the assumption that there were one-way permutations and
indinstinguishability obfuscators for problems in $\cc{P_{/poly}}$.
Another recent result showed that the search version of the Colorful
Carath\'eodory Theorem is in $\cc{PPAD} \cap \cc{PLS}$, and left open
whether the problem is also in $\CLS$~\cite{colorfulcara2017}.

Unfortunately \CLS is not particularly well-understood, and a glaring
deficiency is the current lack of any complete problem for the class. In their
original paper, Daskalakis and Papadimitriou suggested two natural candidates
for complete problems for \CLS, \ContractionMap and \PLCP, and this remains an
open problem.
Another motivation for studying these two problems is that the problems of
solving Condon's simple stochastic games can be reduced to each of them
(separately) in polynomial time
and, in turn, there is sequence of polynomial-time reductions from parity games 
to mean-payoff games to discounted games to simple stochastic 
games~\cite{puri1996theory,gartner2005simple,jurdzinski2008simple,zwick1996complexity,hansen2013complexity}.
The complexity of solving these problems is unresolved and has received 
much attention over many years~(see, for example, 
\cite{zwick1996complexity,condon1992complexity,fearnley2010linear,jurdzinski1998deciding,bjorklund2004combinatorial,fearnley2016complexity}).
In a recent breakthrough, a quasi-polynomial time algorithm for parity games
was presented~\cite{parity}.
For mean-payoff, discounted, and simple stochastic games, the best-known 
algorithms run in subexponential time~\cite{ludwig1995subexponential}.
The existence of a polynomial time algorithm for solving any of these games
would be a major breakthrough.
For \ContractionMap and \PLCP no subexponential time algorithms 
are known, and providing such algorithms would be a major breakthrough.
As the most general of these problems, and thus most likely to be 
\CLS-hard, we study \ContractionMap and \PLCP.

\noindent
\textbf{Our contribution.}
We make progress towards settling the complexity of both of these problems. In
the problem $\ContractionMap$, as defined in~\cite{daskalakis2011continuous}, we
are asked to find an approximate fixed point of a function $f$ that is purported
to contracting with respect to a metric induced by a norm (where the choice of
norm does not matter but is not part of the input), or to give a violation of
the contraction property for $f$. We introduce a problem, \MMCM, that allows
the specification of a purpoted meta-metric $d$ as part of the input of the problem,
along with the function $f$. We are asked to either find an approximate fixed
point of~$f$, a violation of the contraction property for~$f$ with respect to
$d$, a violation of the Lipschitz continuity of $f$ or $d$, or a witness that $d$ violates 
the meta-metric properties.
We show that \MMCM is \CLS-complete, thus identifying a first natural \CLS-complete problem.
We note that, contemporaneously and independently of our work, Daskalakis, Tzamos, and
Zampetakis~\cite{DTZ17} have defined the problem \MBanach and shown it is in \CLS-complete.
Their \CLS-hardness reduction produces a metric and 
is thus stronger than our \CLS-hardness result for \MMCM. 
We discuss both results in more detail in Section~\ref{sec:MMCMisCLScomplete}.

Our second result is to show that \PLCP can be reduced to \EOML.
The \EOML problem was introduced to capture problems that have a PPAD directed
path structure while that also allow us to keep count of exactly how far the vertex is from
the start of the path. In a sense, this may seem rather unnatural, as many common
problems do not seem to have this property. In particular, while the \PLCP problem
has a natural measure of progress towards a solution given by Lemke's algorithm, 
this is given in the form of a potential function, rather than an exact measure
of the number of steps from the beginning of the algorithm.

To address this, we introduce a new problem \EOPL which captures problems with a
PPAD path structure that also allow have a potential function that decreases
along this path. It is straightforward to show that \EOPL is more general than
\EOML. However, despite its generality, we are also able to show that \EOPL can
be reduced to \EOML in polynomial time, and so the two problems are equivalent
under polynomial time reductions.  We show that \PLCP can be reduced to \EOPL,
which provides an alternative proof that $\PLCP$ is in $\CLS$.

We believe that the \EOPL problem is of independent interest, as it naturally
unifies the circuit-based view of $\cc{PPAD}$ and of $\cc{PLS}$, and is defined
in the spirit of the canonical definitions of $\cc{PPAD}$ and $\cc{PLS}$.  There
are two obvious lines for further research.  Given the reduction we provide,
\EOPL and \EOML, are more likely candidates for $\CLS$-hardness than $\PLCP$. 
Alternatively, one could attempt to reduce \EOPL to \PLCP, thereby showing that
that \PLCP is complete for the complexity class defined by these two problems,
and in doing so finally resolve the long-standing open problem of the complexity
of \PLCP.
We note that, in the case of finding a Nash equilibrium of a two-player game,
which we now know is
\PPAD-complete~\cite{chen2009settling,daskalakis2009complexity}, the definition
of \PPAD was inspired by the path structure of the Lemke-Howson algorithm, as
our definition of \EOPL is directly inspired by the path structure of Lemke
paths for P-matrix LCPs.







\section{Preliminaries}


In this section, we define polynomial-time reductions between total search problems
and the complexity class \CLS.
\begin{definition}
For total functions problems, a (polynomial-time) reduction from problem~$A$ to
problem $B$ is a pair of polynomial-time functions $(f,g)$, such that $f$ 
maps an instance $x$ of $A$ to an instance $f(x)$ of $B$, and $g$ maps
any solution $y$ of $f(x)$ to a solution $g(y)$ of $x$.
\end{definition}
Following~\cite{daskalakis2011continuous}, we define the complexity class $\CLS$
as the class of problems that are reducible to the following problem \CLO.

\begin{definition}[\CLO~\cite{daskalakis2011continuous}]
\label{def:CLO}
Given two arithmetic circuits computing functions $f : [0,1]^3\to [0,1]^3$ and $p :
[0,1]^3 \to [0,1]$ and parameters $\e, \lambda > 0$, find either:
\begin{enumerate}[leftmargin=*,label=(C\arabic*)]
\item a point $x\in [0,1]^3$ such that $p(x) \leq p(f(x)) - \e$ or \label{c_fixpoint}
\item a pair of points $x,y\in [0,1]^3$ satisfying either \label{c_violation}
  \begin{enumerate}[label=(C\arabic{enumi}\alph*)] 
  \item $\Norm{f(x) - f(y)} > \lambda \Norm{x-y}$ or \label{c_bad_f}
  \item $\Norm{p(x) - p(y)} > \lambda \Norm{x-y}$. \label{c_bad_p}
  \end{enumerate}
\end{enumerate}
\end{definition}

In Definition~\ref{def:CLO}, $p$ should be thought of as a \emph{potential}
function, and $f$ as a \emph{neighbourhood} function that gives a candidate
solution with better potential if one exists. Both of these functions are 
purported to be Lipschitz continuous. A solution to the problem is either an approximate
potential minimizer or a witness for a violation of Lipschitz continuity.

\begin{definition}[\CM~\cite{daskalakis2011continuous}]
We are given as input an arithmetic circuit computing $f: [0,1]^3\to [0,1]^3$,
a choice of norm $\Norm{\cdot}$, constants \mbox{$\e,c \in (0,1)$},
and $\delta > 0$, and we are promised that $f$ is $c$-contracting w.r.t. $\Norm{\cdot}$.
The goal is to find
\begin{enumerate}[label=(CM\arabic*)]
\item a point $x\in [0,1]^3$ such that $d(f(x),x) \leq \delta$, 
\item or two points $x,y\in [0,1]^3$ such that $\Norm{f(x) - f(y)}/\Norm{x-y} > c$. 
\end{enumerate}
\end{definition}

In other words, the problem asks either for an approximate fixed point of $f$ or
a violation of contraction. As shown in~\cite{daskalakis2011continuous}, \CM is
easily seen to be in \CLS by creating instances of \CLO with $p(x) =
\Norm{f(x)-x}$, $f$ remains as $f$, Lipschitz constant $\lambda = c+1$, and $\epsilon =
(1-c)\delta$.

\section{\MMCM is CLS-Complete}
\label{sec:MMCMisCLScomplete}

In this section, we define \MMCM and show that it is \CLS-complete.
In a \emph{meta-metric}, all the requirements of a metric are satisfied except
that the distance between identical points is not necessarily zero. The
requirements for $d$ to be a meta-metric are given in the following definition.

\begin{definition}[Meta-metric]
\label{def:metametric}
Let $\mathcal{D}$ be a set and $d:\mathcal{D}^2 \mapsto \Real$ a function such that:
\begin{enumerate}
\item $d(x, y) \ge 0$;
\item $d(x, y) = 0$ implies $x = y$ (but, unlike for a metric, the converse is not required);
\item $d(x, y) = d(y, x)$;
\item $d(x, z) \le d(x, y) + d(y, z)$.
\end{enumerate}
Then $d$ is a meta-metric on $\mathcal{D}$.
\end{definition}

The problem \CM, as defined in~\cite{daskalakis2011continuous}, was inspired by
Banach's fixed point theorem, where the contraction can be with respect to any
metric.  In~\cite{daskalakis2011continuous}, for \CM the assumed metric was any
metric induced by a norm. The choice of this norm (and thus metric) was
considered part of the definition of the problem, rather than part of the
problem input. In the following definition of \MMCM, the contraction is with
respect to a meta-metric, rather than a metric, and this meta-metric is given as part of the input of
the problem.

\begin{definition}[\MMCM]
\label{def:MMCM}
We are given as input an arithmetic circuit computing $f: [0,1]^3\to [0,1]^3$,
an arithmetic circuit computing a meta-metric $d : [0,1]^3\times [0,1]^3 \to
[0,1]$, some $p$-norm $\Norm{\cdot}_r$ and constants \mbox{$\e,c \in (0,1)$}
and $\delta > 0$, and we are promised that $f$ is $c$-contracting with
respect to $d$, and $\lambda$-continuous with respect to $\Norm{\cdot}$, and
that $d$ is $\gamma$-continuous with respect to $\Norm{\cdot}$. The goal is
to find
\begin{enumerate}[label=(M\arabic*)]
\item a point $x\in [0,1]^3$ such that $d(f(x),x) \leq \e$, \label{m_fixpoint}
\item or two points $x,y\in [0,1]^3$ such that \label{m_violation}
  \begin{enumerate}[label=(M\arabic{enumi}\alph*)]
    \item $d(f(x),f(y))/d(x,y) > c$, \label{m_not_contracting}
    \item $\Norm{d(x,y) - d(x',y')}/\Norm{(x,y)-(x',y')} > \delta$, or \label{m_bad_metametric}
    \item $\Norm{f(x) - f(y)}/\Norm{x-y} > \lambda$. \label{m_bad_f}
  \end{enumerate}
\item points $x,y$, or $x,y,z$ in $[0,1]^3$ that witness a violation of one 
	of the four defining properties of a meta-metric (Definition~\ref{def:metametric}). \label{mm_violation}
\end{enumerate}
\end{definition}

\begin{definition}[\GCM]
\label{def:GCM}
The definition is identical to that of Definition~\ref{def:MMCM} identical except for the fact that 
solutions of type \ref{mm_violation} are not allowed.
\end{definition}

So, while \MMCM allows violations of $d$ being a meta-metric
as solutions, \GCM does not. 

\begin{theorem}
\label{thm:GCMinCLS}
  \GCM is in \CLS.
\end{theorem}
\begin{proof}
  Given an instance $X=(f,d,\e,c,\lambda,\delta)$ of \GCM, we set $p(x) \triangleq d(f(x),x)$. Then our $\CLO$ instance is the following:
  \[Y=(f, p, \lambda' \triangleq (\lambda + 1) \delta, \e' \triangleq (1-c)\e).\]
Now consider any solution to $Y$. If our solution is of type \ref{c_fixpoint}, a
	point $x$ such that $p(f(x)) > p(x) - \e'$, then we have $d(f(f(x)),f(x))
	> d(f(x),x) - (1-c)\e$, and either $d(f(x),x) \leq \e$, in which case $x$ is
	a solution for $X$, or $d(f(x),x) > \e$. In the latter case, we can divide
	on both sides to get \[ \frac{d(f(f(x)),f(x))}{d(f(x),x)} > 1-
	\frac{(1-c)\e}{d(f(x),x)} \geq 1- (1-c) = c\text{,} \] giving us a violation
	of the claimed contraction factor of $c$, and a solution of type
	\ref{m_not_contracting}.

If our solution is a pair of points $x,y$ of type \ref{c_bad_f} satisfying $\Norm{f(x) - f(y)}/\Norm{x-y} > \lambda' \geq \lambda$, then this gives a violation of the $\lambda$-continuity of $f$. If instead $x,y$ are of type \ref{c_bad_p} so that $\Norm{p(x) - p(y)}/\Norm{x-y} > \lambda'$, then we have
\[ \Abs{d(f(x),x) - d(f(y),y)} = \Abs{p(x) - p(y)} > (\lambda+1)\delta \Norm{x-y}\text{.} \]
We now observe that if
\[ \Abs{d(f(x),x) - d(f(y),y)} \leq \delta (\Norm{f(x)-f(y)} + \Norm{x - y}) \quad \text{and} \Norm{f(x) - f(y)}/\Norm{x-y} \leq \lambda,\] 
	then we would have
\[\Abs{d(f(x),x) - d(f(y),y)} \leq \delta (\Norm{f(x) - f(y)} + \Norm{x-y}) \leq (\lambda + 1)\delta \Norm{x-y},\] 
which contradicts the above inequality, so either the $\delta$ continuity of $d$ must be violated giving a solution to $X$ of type \ref{m_bad_metametric} or the $\lambda$ continuity of $f$ must be violated giving a solution of type \ref{m_bad_f}.
Thus we have shown that $\GCM$ is in $\CLS$.
\end{proof}

Now that we have shown that \GCM is total, we note 
that since the solutions of \GCM are a subset of those for \MMCM, we have the following.

\begin{observation}
\label{obs:MMCMtoGCM}
\MMCM can be reduced in polynomial-time to \GCM.
\end{observation}

Thus, by Theorem~\ref{thm:GCMinCLS}, we have that \MMCM is in \CLS.
Next, we show that \MMCM is \CLS-hard by a reduction 
from the canonical \CLS-complete problem \CLO to an instance of \MMCM.
By Observation~\ref{obs:MMCMtoGCM}, we then also have that \GCM is \CLS-hard.

\begin{theorem}
\label{thm:MMCMisCLShard}
  $\MMCM$ is $\CLS$-hard.
\end{theorem}
\begin{proof}
  Given an instance $X = (f,p,\e,\lambda)$ of \CLO, we construct a meta-metric $d(x,y) = p(x) + p(y) + 1$. 
	Since $p$ is non-negative, $d$ is non-negative, and by construction, $d$ is symmetric and satisfies the triangle inequality. Finally, $d(x,y) > 0$ for all choices of $x$ and $y$ so $d$ is a valid meta-metric (Definition~\ref{def:metametric}) Furthermore, if $p$ is $\lambda$-continuous with respect to the given $p$-norm $\Norm{\cdot}_r$, then $d$ is ($2^{1/r-1}\lambda$)-continuous with respect to $\Norm{\cdot}_r$. For clarity, in the below proof we'll omit the subscript $r$ when writing the norm of an expression. To see this we observe that $x,x',y,y'\in [0,1]^n$, we have $\Norm{p(x)-p(x')}/\Norm{x-x'} \leq \lambda$ and $\Norm{p(y) - p(y')}/\Norm{y-y'} \leq \lambda$, so
  \begin{align*}
    \frac{\Norm{d(x,y) - d(x',y')}}{\Norm{(x,y) - (x',y')}}
    &= \frac{\Norm{p(x) - p(x') + p(y) - p(y') + 1 - 1}}{\Norm{(x,y) -(x',y')}}\ \leq\ \frac{\lambda\Norm{x-x'} + \lambda\Norm{y-y'}}{\Norm{(x,y) -(x',y')}}\\
    &\leq \frac{\lambda\Norm{x-x'} + \lambda\Norm{y-y'}}{2^{1-1/r}(\Norm{x-x'} + \Norm{y-y'})} \leq 2^{1/r-1}\lambda \text{.}
  \end{align*}
We'll output an instance $Y = (f,d,\e'=\e,c = 1-\e/4,\delta=\lambda, \lambda'=2^{1/r-1}\lambda)$.

Now we consider solutions for the instance $Y$ and show that they correspond to solutions for our input instance $X$.
First, we consider a solution of type \ref{m_fixpoint}, a point $x\in [0,1]^3$ such that $d(f(x),x) \leq \e'=\e$. We have $p(f(x)) + p(x) + 1 \leq \e$, but this can't happen since $\e < 1$ and $p$ is non-negative, so solutions of this type cannot exist.

Now consider a solution that is a pair of points $x,y\in [0,1]^3$ satisfying one of the conditions in \ref{m_violation}. If the solution is of type \ref{m_not_contracting}, we have $d(f(x),f(y)) > c d(x,y)$, and by our choice of $c$ this is exactly
\[\frac{d(f(x),f(y))}{d(x,y)} > (1-\e/4)\] and
\begin{align*}
  p(f(x)) + p(f(y)) + 1 &> (1-\e/4) (p(x) + p(y) + 1)\\
                        &\geq p(x) + p(y) - 3\e/4
\end{align*}
so either $p(f(x)) > p(x) - \e$ or $p(f(y)) > p(y) - \e$, and one of $x$ or $y$ must be a fixpoint solution to our input instance.
Solutions of type \ref{m_bad_metametric} or \ref{m_bad_f} immediately give us violations of the $\lambda$-continuity of $f$, and thus solutions to $X$.

This completes the proof that $\MMCM$ is $\CLS$-hard.
\end{proof}

So combining these results we have the following.

\begin{theorem}
\MMCM and \GCM are \CLS-complete.
\end{theorem}

Finally, as mentioned in the introduction, we note the following.
Contemporaneously and independently of our work, Daskalakis, Tzamos, and
Zampetakis~\cite{DTZ17} defined the problem \MBanach, which is like \MMCM except
that it requires a metric, as opposed to a meta-metric.  They show that \MBanach
is \CLS-complete.  Since every metric is a meta-metric, \MBanach can be
trivially reduced in polynomial-time to \MMCM. Thus, their \CLS-hardness result
is stronger than our Theorem~\ref{thm:MMCMisCLShard}.
The
containment of \MBanach in \CLS is implied by the containment of \MMCM in \CLS. 
To prove
that \MMCM is in \CLS, we first reduce to \GCM, which we then show is in \CLS.
Likewise, the proof in~\cite{DTZ17} that \MBanach is in \CLS works even when
violations of the metric properties are not required as solutions, so they, like
us, actually show that \GCM is in \CLS.

\section{\EOML to \EOPL and Back}
\label{sec:EOMLtoEOPL}

In this section, we define a new problem \EOPL.
Then, we design polynomial-time reductions from \EOML to \EOPL, and
from \EOPL to \EOML, thereby showing that the two problems 
are polynomial-time equivalent. In Section~\ref{sec:PLCPtoEOPL},
we reduce \PLCP to \EOPL.

First we recall the definition of \EOML, which was
first defined in~\cite{hubavcek2017hardness}.
It is close in spirit to the problem \EOL that is used
to define \PPAD~\cite{papadimitriou1994complexity}. 

\begin{definition}[\EOML~\cite{hubavcek2017hardness}]
Given circuits $S,P: \{0,1\}^n \rightarrow \{0,1\}^n$, and $V:\{0,1\}^n\rightarrow \{0,\dots, 2^n\}$ such that $P(0^n) =0^n\neq S(0^n)$ and $V(0^n)=1$, find a string $\xx \in \{0,1\}^n$ satisfying one of the following 
\begin{enumerate}[label=(T\arabic*)]
\item either $S(P(\xx))\neq \xx \neq 0^n$ or $P(S(\xx))\neq \xx$,
\item $\xx\neq 0^n, V(\xx)=1$,
\item either $V(\xx)>0$ and $V(S(\xx))-V(\xx)\neq 1$, or $V(\xx)>1$ and $V(\xx)-V(P(\xx))\neq 1$. 
\end{enumerate}
\end{definition}
Intuitively, an \EOML is an \EOL instance that is also equipped with an
``odometer'' function. The circuits $P$ and $S$ implicitly define an
exponentially large graph in which each vertex has degree at most 2, just as in \EOL, and condition T1 says that the end of
every line (other than $0^n$) is a solution.
In particular, the string 
$0^n$ is guaranteed to be the end of a line, and so a solution can be found by
following the line that starts at $0^n$.
 The function $V$ is intended to help with this, by giving the number of steps
that a given string is from the start of the line. We have that $V(0^n) = 1$,
and that $V$ increases by exactly 1 for each step we make along the line.
Conditions T2 and T3 enforce this by saying that any violation of the property
is also a solution to the problem.

In \EOML, the requirement of incrementing $V$ by exactly one as walk along the
line is quite restrictive. We define a new problem, \EOPL,  which is similar in
spirit to \EOL, but drops the requirement of always incrementing the potential
by one as we move along the line.

\begin{definition}[\EOPL]
\label{def:EOPL}
Given Boolean circuits $S,P : \Set{0,1}^n \to \Set{0,1}^n$ such that $P(0^n) =0^n\neq S(0^n)$ and a Boolean circuit $V: \Set{0,1}^n \to \Set{0,1,\dotsc,2^m - 1}$ such that $V(0^n) = 0$ find one of the following:
\begin{enumerate}[label=(R\arabic*)]
\item A point $x \in \Set{0,1}^n$ such that $S(P(x)) \neq x \neq 0^n$ or $P(S(x)) \neq x$.
\item A point $x \in \Set{0,1}^n$ such that $x \neq S(x)$, $P(S(x)) = x$, and $V(S(x)) - V(x) \leq 0$.
\end{enumerate}
\end{definition}

The key difference here is that the function $V$ is required to be strictly
monotonically increasing as we walk along the line, but the amount that it
increases in each step is not specified.
At first glance, the definition of \EOPL may seem more general and more likely to 
capture the whole class \CLS. In fact, we will show that \EOML and \EOPL are 
inter-reducible in polynomial-time.
\begin{theorem}
\EOML and \EOPL are equivalent under polynomial-time reductions.
\end{theorem}
As expected, the reduction from~\EOML to \EOPL is relatively easy. It requires
handling the difference in potential at $0^n$ and vertices with potential zero that
are not discarded directly as possible solutions in \EOPL. We make the latter
self loops, but that creates extra starts and ends of lines which need to be
handled. Full details of the reduction with proofs are in
Appendix~\ref{sec:EOMLtoEOPL}.

The reduction from \EOPL to \EOML is involved, and appears in detail in
Appendix~\ref{sec:eopl2eoml}. Here the basic idea is to insert missing single
increments in between by introducing new vertices along the original edges. To
allow this we need to encode potential itself in the vertex description. If
there is an edge from $\uu$ to $\uu'$ in the \EOPL instance whose respective
potentials are $p$ and $p'$ such that say $p<p'$ then we create edges $(u,p)\ra
(u,p+1)\ra \dots \ra (u,p'-1)\ra (u,p')$. However, this creates a lot of dummy
vertices, namely those that never appear on any edge due to irrelevant potential
values, i.e., in this example $(u,\pi)$ with $\pi <p$ or $\pi\ge p'$. We make
them self loops (not an end-of-line) with zero potential, and since
non-end-of-line solutions of \EOML, namely $T2$ and $T3$, must have strictly
positive potential, these will never create a solution of the \EOML instance.

In addition, a number of issues need to be handled with consistency: $(a)$
a $T2$ type solution of \EOML may be neither at the end of any line nor be a 
potential violation in \EOPL; we do extra (linear time) work to handle such
solutions, $(b)$ a $T3$ type potential violation may not be on a ``valid'' edge as
required by \EOPL. $(c)$ ``invalid'' edges, $(d)$ potential difference at the
initial vertex $0^n$, etc.




\section{Reduction from \PLCP to \EOPL}
\label{sec:PLCPtoEOPL}

In this section we present a polynomial-time reduction from the P-matrix Linear
Complementarity Problem (\PLCP) to \EOPL.
A Linear Complementarity Problem (LCP) is defined as follows. Now on by $[n]$ we mean set $\{1,\dots,n\}$.

\begin{definition}[LCP]
\label{def:lcp}
Given a matrix $M \in R^{d \times d}$ and a vector $\qq\in \Real^{d\times 1}$,
find a vector~{$\yy\in \Real^{d \times 1}$} such that:
\begin{equation}\label{eq:lcp}
M\yy\le \qq;\ \ \ \ \yy\ge 0;\ \ \ \ y_i(\qq - M\yy)_i =0,\ \forall i \in [n].
\end{equation}
\end{definition}
In general, an LCP may have no solution, and deciding whether one does is
\NP-complete~\cite{chung1989np}. If the matrix $M$ is a P-matrix, as defined
next, then the LCP $(M,\qq)$ has a unique solution for all $\qq\in \Real^{d\times
1}$.
\begin{definition}[P-matrix]
\label{def:Pmatrix}
A matrix $M \in \Real^{d \times d}$ is called a P-matrix if every principle
minor of $M$ is positive, {\em i.e.,} for every subset $S\subseteq[d]$, the
sub-matrix $N=[M_{i,j}]_{i\in S, j\in S}$ has strictly positive determinant. 
\end{definition}
In order to define a problem that takes all matrices $M$ as input without 
a promise, Megiddo~\cite{megiddo1988note} defined \PLCP as the following problem
(see also~\cite{megiddo1991total}).
\begin{definition}[\PLCP] \label{def:plcp} Given a matrix $M\in \Real^{d\times
d}$ and a vector $\qq\in \Real^{d\times 1}$, either:
\begin{enumerate}[label=(Q\arabic*)] \item Find vector $\yy \in \Real^{n
			\times 1}$ that satisfies (\ref{eq:lcp}) \item Produce a witness
that $M$ is a not a P-matrix, {\em i.e.,} find $S\subset [d]$ such that for
submatrix $N=[M_{i,j}]_{i\in S, j\in S}$, $det(N)\le 0$.  \end{enumerate}
\end{definition}
Later, Papadimitriou showed that \PLCP is in
\PPAD~\cite{papadimitriou1994complexity}, and then Daskalakis and Papadimitrou
showed that it is in \CLS~\cite{daskalakis2011continuous} (based on the
potential reduction method in~\cite{kojima1992interior}).  Designing a
polynomial-time solution for the \PLCP problem has been open for decades, at
least since the 1978 paper of Murty~\cite{murty1978computational} that provided
exponential-time examples for \emph{complementary pivoting algorithms}, such as 
\emph{Lemke's algorithm}~\cite{lemke1965bimatrix}, for
P-matrix Linear Complementarity Problems. Murty's family of P-matrices were
based on the Klee-Minty's cubes that had been used to give exponential-time
examples for the simplex method, and which inspired the research that led to
polynomial-time algorithms for Linear Programming. No similar polynomial-time
algorithms are known for \PLCP though.

Lemke's algorithm introduces an extra variable, say $z$, to the LCP polytope,
and follows a path on the $1$-skeleton of the new polytope (like the simplex 
method for linear programming) based
on complementary pivot rule (details below).  A general LCP need not have a
solution, and thus Lemke's algorithm is not guaranteed to terminate with a
solution.  However, for P-matrix LCPs, Lemke's algorithm terminates.  Indeed, if
Lemke's algorithm does not terminate with a solution, it provides a witness that
the matrix $M$ is not a P-matrix.  The structure of the path traced by Lemke's
algorithm is crucial for our reduction, so let us first briefly describe the
algorithm.

\subsection{Lemke's Algorithm}
\label{sec:lemke}

The explanation of Lemke's algorithm in this section is taken from \cite{GMSV}.
The problem is interesting only when $\qq \not \geq 0$, since otherwise $\yy = 0$ is a trivial solution. Let us introduce
slack variables $\ps$ to obtain the following equivalent formulation:
\begin{equation} \label{eq:b} \MM \yy  + \ps = \pq, \ \ \ \  \yy \geq 0, \ \ \ \ \ps \geq 0 \ \ \ \ \mbox{and} \ \ \ \ y_is_i = 0,\ \forall i\in[d].  \end{equation}

Let $\CQ$ be the polyhedron in $2d$ dimensional space defined by the first three conditions; we will assume that $\CQ$ is
non-degenerate (just for simplicity of exposition; this will not matter for our reduction).  
Under this condition, any solution to (\ref{eq:b}) will be a vertex of $\CQ$, since it must satisfy $2d$
equalities. Note that the set of solutions may be disconnected.
An ingenious idea of Lemke was to introduce a new variable and consider the system:
\begin{equation} \label{eq:c} \MM \yy  + \ps -z \one  = \pq, \ \ \ \  \yy \geq 0, \ \ \ \ \ps \geq 0, \ \ \ \  z \geq 0  \ \
\ \ \mbox{and} \ \ \ \ y_is_i = 0,\ \forall i\in[d].  \end{equation}
The next lemma follows by construction of (\ref{eq:c}).
\begin{lemma}\label{lem:lemke1}
Given $(\MM,\qq)$, $(\yy,\ps,z)$ satisfies \eqref{eq:c} with $z=0$ iff $\yy$ satisfies~\eqref{eq:lcp}.
\end{lemma}
Let $\CPol$ be the polyhedron in $2d + 1$ dimensional space defined by the first four conditions of \eqref{eq:c}, i.e.,
\begin{equation}\label{eq:cp}
\CPol = \{ (\yy,\ps, z) \ |\ \MM \yy  + \ps -z \one  = \pq, \ \ \ \yy \geq 0, \ \ \ \ps \geq 0, \ \ \  z \geq 0\};
\end{equation}
we will assume that $\CPol$ is {\em non-degenerate}.  

Since any solution to (\ref{eq:c}) must still satisfy $2d$ equalities in $\CPol$, the set of solutions, say
$S$, will be a subset of the one-skeleton of $\CPol$, i.e., it will consist of edges and vertices of $\CPol$.  Any solution to
the original system (\ref{eq:b}) must satisfy the additional condition $z = 0$ and hence will be a vertex of $\CPol$.

Now $S$ turns out to have some nice properties. Any point of $S$ is {\em fully labeled} in the sense that for each $i$, $y_i
= 0$ or $s_i = 0$.  We will say that a point of $S$ {\em has duplicate label i} if $y_i = 0$ and $s_i = 0$ are both satisfied
at this point. Clearly, such a point will be a vertex of $\CPol$ and it will have only one duplicate label.  Since there are
exactly two ways of relaxing this duplicate label, this vertex must have exactly two edges of $S$ incident at it.  Clearly, a
solution to the original system (i.e., satisfying $z = 0$) will be a vertex of $\CPol$ that does not have a duplicate label.  On
relaxing $z=0$, we get the unique edge of $S$ incident at this vertex.

As a result of these observations, we can conclude that $S$ consists of paths and cycles.  Of these paths, Lemke's algorithm
explores a special one.  An unbounded edge of $S$ such that the vertex of $\CPol$ it is incident on has $z > 0$ is called a
{\em ray}.  Among the rays, one is special -- the one on which $\yy = 0$. This is called the {\em primary ray} and the rest
are called {\em secondary rays}. Now Lemke's algorithm explores, via pivoting, the path starting with the primary ray. This
path must end either in a vertex satisfying $z = 0$, i.e., a solution to the original system, or a secondary ray. In the
latter case, the algorithm is unsuccessful in finding a solution to the original system; in particular, the original system
may not have a solution.  
We give the full pseudo-code for Lemke's algorithm in Appendix~\ref{app:lemke}.



\subsection{Polynomial time reduction from \PLCP to \EOPL}

It is well known that if matrix $M$ is a P-matrix (\PLCP), then $z$ strictly
decreases on the path traced by Lemke's algorithm \cite{cottle2009linear}.
Furthermore, by a result of Todd~\cite[Section 5]{todd1976orientation}, paths traced by
complementary pivot rule can be locally oriented.  Based on these two facts, 
we now derive a polynomial-time reduction from \PLCP to \EOPL.

Let $\CI=(M,\qq)$ be a given \PLCP instance, and let $\CL$ be the length of the 
bit representation of $M$ and $\qq$. 
We will reduce $\CI$ to an \EOPL instance $\CE$ in time $poly(\CL)$. 
According to Definition~\ref{def:EOPL}, the instance $\CE$ is defined 
by its vertex set $\vert$, and procedures $S$ (successor), $P$ (predecessor) and $\pot$ (potential). 
Next we define each of these. 

As discussed in Section \ref{sec:lemke} the linear constraints of (\ref{eq:c})
on which Lemke's algorithm operates forms a polyhedron $\CPol$ given in
(\ref{eq:cp}). We assume that $\CPol$ is non-degenerate. This is without
loss of generality since, a typical way to ensure this is by perturbing $\qq$ so
that configurations of solution vertices remain unchanged
\cite{cottle2009linear}, and since $M$ is unchanged the LCP is still \PLCP. 

Lemke's algorithm traces a path on feasible points of (\ref{eq:c}) which is on
$1$-skeleton of $\CPol$ starting at $(\yy^0,\ps^0,z^0)$, where:
\begin{equation}\label{eq:v0}
\yy^0=0,\ \ \ \ \ z^0= |\min_{i \in [d]} q_i|,\ \ \ \ \  \ps^0=\qq+z\ones
\end{equation}
We want to capture
vertex solutions of (\ref{eq:c}) as vertices in \EOPL instance $\CE$. To
differentiate we will sometimes call the latter {\em configurations}. Vertex
solutions of (\ref{eq:c}) are exactly the vertices of polyhedron $\CPol$ with
either $y_i=0$ or $s_i=0$ for each $i\in [d]$. Vertices of (\ref{eq:c}) with
$z=0$ are our final solutions (Lemma \ref{lem:lemke1}). While each of its {\em
non-solution} vertex has a duplicate label. Thus, a vertex of this path can be
uniquely identified by which of $y_i=0$ and $s_i=0$ hold for each $i$ and its
duplicate label. This gives us a representation for vertices in the \EOPL
instance $\CE$. 

\medskip

\noindent{\bf \EOPL Instance $\CE$.}
\begin{itemize}
\item Vertex set $\vert=\{0,1\}^n$ where $n = 2d$. 
\item Procedures $S$ and $P$ as defined in Tables \ref{tab:S} and \ref{tab:P} respectively
\item Potential function $\pot:\vert \rightarrow \{0,1,\dots, 2^m-1\}$ defined in Table \ref{tab:F} for $m=\lceil ln(2\Delta^3)\rceil$, 
	  where $$\Delta=(n! \cdot I_{max}^{2d+1})+1$$ 
	  and $I_{max} = \max\{\max_{i,j\in [d]} M(i,j),\ \max_{i\in [d]} |q_i|\}$. 
\end{itemize}

For any vertex $\uu\in \vert$, the first $d$ bits of $\uu$ represent
which of the two inequalities, namely $y_i\ge 0$ and $s_i\ge 0$, are tight for
each $i \in [d]$. A valid setting of the second set of $d$ bits will have 
at most one non-zero bit -- if none is one then $z=0$, otherwise the location of one bit indicates the duplicate label. 
Thus, there are many invalid configurations, namely
those with more than one non-zero bit in the second set of $d$ bits. 
These are dummies that we will handle separately, and we define a procedure 
$\isvalid$ to identify non-dummy vertices in Table \ref{tab:iv} (in Appendix \ref{app:proc}). 
To go between ``valid'' vertices of $\CE$ and corresponding vertices of the Lemke polytope
$\CPol$ of LCP $\CI$, we define procedures $\eti$ and $\ite$ in Table
\ref{tab:ei} (in Appendix \ref{app:proc}).  
By construction of $\isvalid$, $\eti$ and $\ite$, the next lemma follows.
All the missing proofs of this section may be found in Appendix \ref{app:PLCPtoEOPL}.

\begin{lemma}\label{lem:vert}
If $\isvalid(\uu)=1$ then $\uu=\ite(\eti(\uu))$, and the corresponding vertex $(\yy,\ps,z)\in \eti(\uu)$ of $\CPol$ is feasible in (\ref{eq:c}). If $(\yy,\ps,z)$ is a feasible vertex of (\ref{eq:c}) then $\uu=\ite(\yy,\ps,z)$ is a valid configuration, {\em i.e.,} $\isvalid(\uu)=1$.
\end{lemma}

\begin{figure}[htbp]
   \centering
\vspace{-1cm}

  \includegraphics[width=\textwidth]{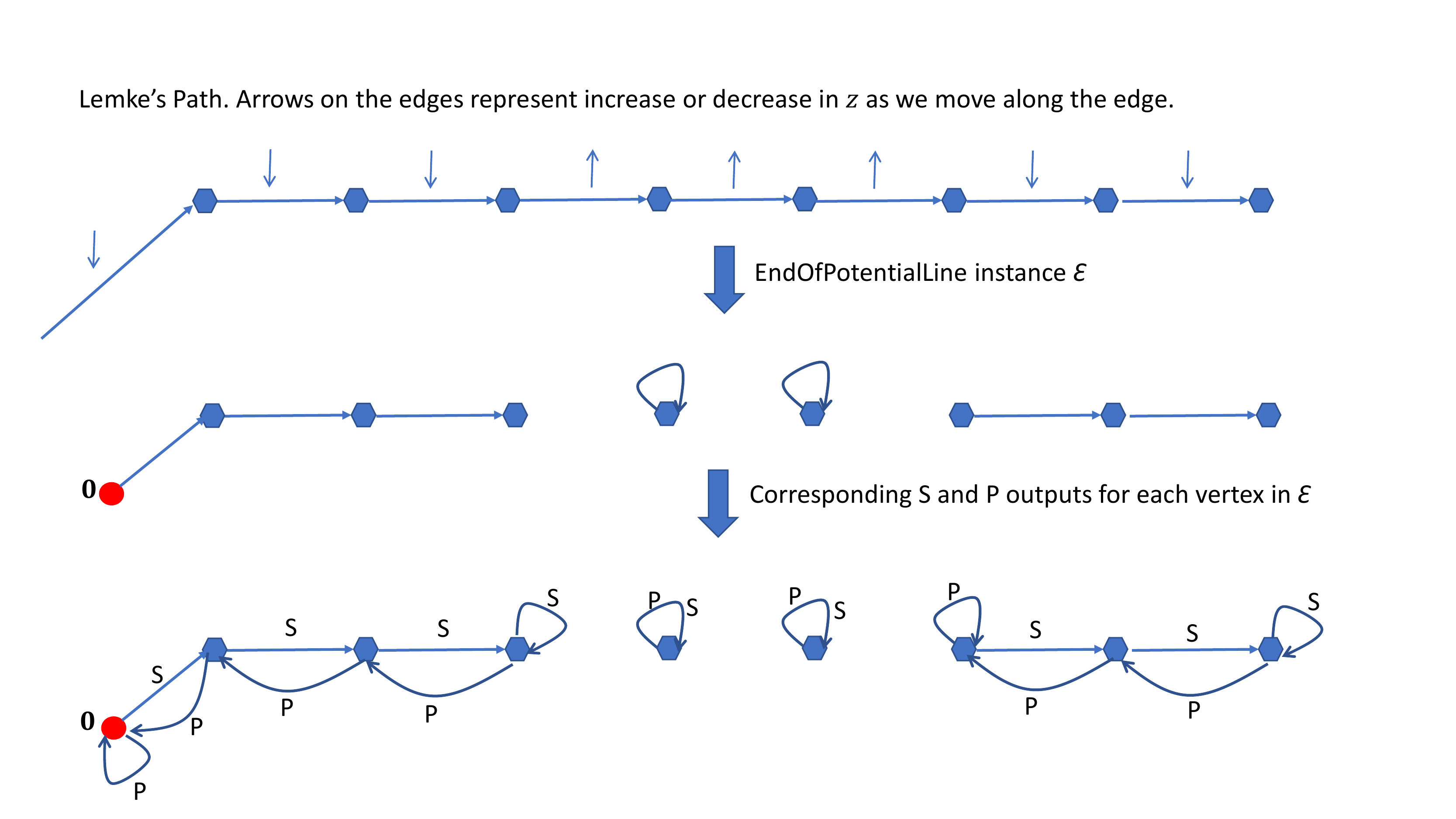}
 \caption{Construction of $S$ and $P$ for \EOPL instance $\CE$ from the Lemke
 	path. The first path is the Lemke path and the arrows on its edges indicate
 	whether the value of $z$ increases or decreases along the edge.  Note that
 	the end or start of a path in $\CE$, which is an intermediate vertex in
 	Lemke path that has either decreased and then increased, or increased and then
 	decreased in the value of $z$, is a violation of $M$ being a $P$ matrix
 	\cite{cottle2009linear}, i.e., $\PLt$ type solution of
 	\PLCP.}
 	\label{fig:path}
\end{figure}

The main idea behind procedures $S$ and $P$, given in Tables \ref{tab:S} and
\ref{tab:P} respectively, is the following (also see Figure \ref{fig:path}):
Make dummy configurations in $\vert$ to point to themselves with cycles of
length one, so that they can never be solutions. 
The starting vertex $0^n \in \vert$ points to the configuration that corresponds
to the first vertex of the Lemke path, namely $\uu^0=\ite(\yy^0,\ps^0,z^0)$. 
Precisely, $S(0^n)=\uu^0$, $P(\uu^0)=0^n$ and $P(0^n)=0^n$ (start of
a path). 

For the remaining cases, let $\uu\in \vert$ have corresponding representation
$\xx=(\yy,\ps,z)\in \CPol$, and suppose $\xx$ has a duplicate label. As one
traverses a Lemke path for a P-LCPs, the value of $z$ monotonically decreases.
So, for $S(\uu)$ we compute the adjacent vertex $\xx'=(\yy',\ps',z')$ of $\xx$
on Lemke path such that the edge goes from $\xx$ to $\xx'$, and if the $z'<z$,
as expected, then we point $S(\uu)$ to configuration of $\xx'$ namely
$\ite(\xx')$. Otherwise, we let $S(\uu)=\uu$. Similarly, for $P(\uu)$, we find
$\xx'$ such that edge is from $\xx'$ to $\xx$, and then we let $P(\uu)$ be
$\ite(\xx')$ if $z'>z$ as expected, otherwise $P(\uu)=\uu$. 

For the case when $\xx$ does not have a duplicate label, then we have $z=0$. This is
handled separately since such a vertex has exactly one incident edge on the Lemke
path, namely the one obtained by relaxing $z=0$. According to the direction of 
this edge, we do similar process as before. For example, if the edge goes from 
$\xx$ to $\xx'$, then, if $z'<z$, we set $S(\uu)=\ite(\xx')$ else $S(\uu)=\uu$,
and we always set $P(\uu)=\uu$.  In case the edge goes from $\xx'$ to $\xx$, we
always set $S(\uu)=\uu$, and we set $P(\uu)$ depending on whether or not $z'>z$.

%
\medskip

The potential function $\pot$, formally defined in Table \ref{tab:F},
gives a value of zero to dummy vertices and the starting vertex $0^n$. To all
other vertices, essentially it is $((z^0-z) * \Delta^2)+1$. Since value of $z$
starts at $z^0$ and keeps decreasing on the Lemke path this value will keep
increasing starting from zero at the starting vertex $0^n$. Multiplication by
$\Delta^2$ will ensure that if $z_1>z_2$ then the corresponding potential values 
will differ by at least one. This is because, since $z_1$ and $z_2$ are 
coordinates of two vertices of polytope $\CPol$, their maximum value is $\Delta$
and their denominator is also bounded above by $\Delta$. Hence $z_1-z_2\le
1/\Delta^2$ (Lemma \ref{lem:pot}).  

To show correctness of the reduction we need to show two things: $(i)$ All the procedures are well-defined and polynomial time. $(ii)$ We can construct a solution of $\CI$ from a solution of $\CE$ in polynomial time.

\begin{table}
\begin{minipage}{0.73\textwidth}
\caption{Successor Procedure $S(\uu)$}\label{tab:S}
\begin{tabular}{|l|}
\hline
\hspace{0pt}{\bf If} $\isvalid(\uu)=0$ {\bf then} {\bf Return} $\uu$\\
\hspace{0pt}{\bf If} $\uu=0^n$ {\bf then} {\bf Return} $\ite(\yy^0,\ps^0,z^0)$\\
\hspace{0pt}$\xx=(\yy,\ps,z) \leftarrow \eti(\uu)$\\
\hspace{0pt}{\bf If} $z=0$ {\bf then} \\
\hspace{5pt} $\xx^1\leftarrow$ vertex obtained by relaxing $z=0$ at $\xx$ in $\CPol$. \\
\hspace{5pt} {\bf If} Todd \cite{todd1976orientation} prescribes edge from $\xx$ to $\xx^1$ \\
\hspace{10pt} {\bf then} set $\xx'\leftarrow \xx^1$. {\bf Else Return} $\uu$ \\
\hspace{0pt}{\bf Else} set $l\leftarrow $ duplicate label at $\xx$\\
\hspace{5pt} $\xx^1\leftarrow $ vertex obtained by relaxing $y_l=0$ at $\xx$ in $\CPol$ \\
\hspace{5pt} $\xx^2\leftarrow $ vertex obtained by relaxing $s_l=0$ at $\xx$ in $\CPol$ \\
\hspace{5pt} {\bf If} Todd \cite{todd1976orientation} prescribes edge from $\xx$ to $\xx^1$ \\
\hspace{10pt} {\bf then} $\xx'=\xx^1$ \\
\hspace{5pt} {\bf Else} $\xx'=\xx^2$\\
\hspace{0pt}Let $\xx'$ be $(\yy',\ps',z')$. \\
\hspace{0pt}{\bf If} $z>z'$ {\bf then} {\bf Return} $\ite(\xx')$. {\bf Else} {\bf Return} $\uu$.\\
\hline
\end{tabular}
\end{minipage}%
\hspace{-1cm}
\begin{minipage}{0.23\textwidth}
\caption{Potential Value $\pot(\uu)$}\label{tab:F}
\begin{tabular}{|l|}
\hline
\hspace{0pt} {\bf If} $\isvalid(\uu)=0$ \\
\hspace{5pt} {\bf then} {\bf Return} $0$\\
\hspace{0pt} {\bf If} $\uu=0^n$\\
\hspace{5pt}  {\bf then} {\bf Return} $0$\\
\hspace{0pt} $(\yy,\ps,z) \leftarrow \eti(\uu)$\\
\hspace{0pt} {\bf Return} $\lfloor \Delta^2*(\Delta -z)\rfloor$\\
\hline
\end{tabular}
\end{minipage}
\end{table}

\begin{table}[!htb]
\caption{Predecessor Procedure $P(\uu)$}\label{tab:P}
\begin{tabular}{|l|}
\hline
\hspace{0pt} {\bf If} $\isvalid(\uu)=0$ {\bf then} {\bf Return} $\uu$\\
\hspace{0pt} {\bf If} $\uu=0^n$ {\bf then} {\bf Return} $\uu$\\
\hspace{0pt} $(\yy,\ps,z) \leftarrow \eti(\uu)$\\
\hspace{0pt} {\bf If} $(\yy,\ps,z)=(\yy^0,\ps^0,z^0)$ {\bf then} {\bf Return} $0^n$\\
\hspace{0pt} {\bf If} $z=0$ {\bf then} \\
\hspace{5pt} $\xx^1\leftarrow$ vertex obtained by relaxing $z=0$ at $\xx$ in $\CPol$. \\
\hspace{5pt} {\bf If} Todd \cite{todd1976orientation} prescribes edge from $\xx^1$ to $\xx$ {\bf then} set $\xx'\leftarrow \xx^1$. {\bf Else Return} $\uu$\\
\hspace{0pt} {\bf Else}\\
\hspace{5pt} $l\leftarrow $ duplicate label at $\xx$\\
\hspace{5pt} $\xx^1\leftarrow $ vertex obtained by relaxing $y_l=0$ at $\xx$ in $\CPol$ \\
\hspace{5pt} $\xx^2\leftarrow $ vertex obtained by relaxing $s_l=0$ at $\xx$ in $\CPol$ \\
\hspace{5pt} {\bf If} Todd \cite{todd1976orientation} prescribes edge from $\xx^1$ to $\xx$ {\bf then} $\xx'=\xx^1$ {\bf Else} $\xx'=\xx^2$\\
\hspace{0pt} Let $\xx'$ be $(\yy',\ps',z')$. {\bf If} $z<z'$ {\bf then} {\bf Return} $\ite(\xx')$. {\bf Else} {\bf Return} $\uu$.\\
\hline
\end{tabular}
\end{table}

\begin{lemma}\label{lem:PSF}
Functions $P$, $S$ and $\pot$ of instance $\CE$ are well defined, making $\CE$ a valid \EOPL instance. 
\end{lemma}

There are two possible types of solutions of an \EOPL instance. One indicates
the beginning or end of a line, and the other is a vertex with locally optimal
potential (that does not point to itself). 
First we show that the latter case never arise. For this, we need the
next lemma, which shows that potential differences in two adjacent
configurations adheres to differences in the value of $z$ at corresponding
vertices.

\begin{lemma}\label{lem:pot}
Let $\uu \neq \uu'$ be two valid configurations, i.e.,
	$\isvalid(\uu)=\isvalid(\uu')=1$, and let $(\yy,\ps,z)$ and $(\yy',\ps',z')$
	be the corresponding vertices in $\CPol$. Then the following holds: $(i)$
	$\pot(\uu)=\pot(\uu')$ iff $z=z'$. $(ii)$ $\pot(\uu)>\pot(\uu')$ iff $z<z'$.
\end{lemma}
%
%

Using the above lemma, we will next show that instance $\CE$ has no local maximizer. 

\begin{lemma}\label{lem:t}
Let $\uu,\vv \in \vert$ s.t. $\uu\neq \vv$, $\vv=S(\uu)$, and $\uu=P(\vv)$. Then $\pot(\uu)< \pot(\vv)$.
\end{lemma}
\begin{proof}
Let $\xx=(\yy,\ps,z)$ and $\xx'=(\yy',\ps',z')$ be the vertices in polyhedron $\CPol$ corresponding to $\uu$ and $\vv$ respectively. From the construction of $\vv=S(\uu)$ implies that $z'<z$. Therefore, using Lemma \ref{lem:pot} it follows that $\pot(\vv)<\pot(\uu)$.
\end{proof}

Due to Lemma \ref{lem:t} the only type of solutions available in $\CE$ is where $S(P(\uu))\neq \uu$ and $P(S(\uu))\neq \uu$. Next two lemmas shows how to construct solutions of $\CI$ from these. 

\begin{lemma}\label{lem:t1}
Let $\uu \in \vert$, $\uu \neq 0^n$. 
If $P(S(\uu))\neq \uu$ or $S(P(\uu))\neq \uu$, then $\isvalid(\uu)=1$, and for $(\yy,\ps,z)=\eti(\uu)$ if $z=0$ then $\yy$ is a $\PLo$ type solution of \PLCP instance $\CI=(M,\qq)$. 
\end{lemma}
\begin{proof}
By construction, if $\isvalid(\uu) = 0$, then $S(P(\uu))=\uu$ and $P(S(\uu))=\uu$, therefore $\isvalid(\uu)=0$ when $\uu$ has a predecessor or successor different from $\uu$.
Given this, from Lemma \ref{lem:vert} we know that $(\yy,\ps,z)$ is a feasible vertex in (\ref{eq:c}). Therefore, if $z=0$ then using Lemma \ref{lem:lemke1} we have a solution of the LCP (\ref{eq:lcp}), {\em i.e.,} a type $\PLo$ solution of our \PLCP instance $\CI=(\MM,\qq)$.
%
%
\end{proof}

\begin{lemma}\label{lem:t2}
Let $\uu \in \vert$, $\uu \neq 0^n$ such that $P(S(\uu))\neq \uu$ or $S(P(\uu))\neq \uu$, and let $\xx=(\yy,\ps,z)=\eti(\uu)$. 
If $z\neq 0$ then $\xx$ has a duplicate label, say $l$. And for directions $\sigma_1$ and $\sigma_2$ obtained by relaxing $y_l=0$ and $s_l=0$ respectively at $\xx$, we have $\sigma_1(z)*\sigma_2(z)\ge 0$, where $\sigma_i(z)$ is the coordinate corresponding to $z$. 
\end{lemma}
\begin{proof}
From Lemma \ref{lem:t1} we know that $\isvalid(\uu)=1$, and therefore from Lemma \ref{lem:vert}, $\xx$ is a feasible vertex in (\ref{eq:c}).
From the last line of Tables \ref{tab:S} and \ref{tab:P} observe that $S(\uu)$ points to the configuration of vertex next to $\xx$ on Lemke's path only if it has lower $z$ value otherwise it gives back $\uu$, and similarly $P(\uu)$ points to the previous only if value of $z$ increases.


First consider the case when $P(S(\uu))\neq \uu$. Let $\vv=S(\uu)$ and corresponding vertex in $\CPol$ be $(\yy',\ps',z')=\eti(\vv)$. 
If $\vv\neq \uu$, then from the above observation we know that $z'>z$, and in that
case again by construction of $P$ we will have $P(\vv)=\uu$, contradicting
$P(S(\uu))\neq \uu$. Therefore, it must be the case that $\vv=\uu$.
Since $z\neq 0$ this happens only when the next vertex on Lemke path after $\xx$ has
higher value of $z$ (by above observation). As a consequence of $\vv=\uu$, we also have $P(\uu)\neq \uu$. By construction of $P$ this implies for 
$(\yy'',\ps'',z'')=\eti(P(\uu))$, $z''>z$. Putting both together we get 
increase in $z$ when we relax $y_l=0$ as well as when we relax $s_l=0$ at
$\xx$.

For the second case $S(P(\uu))\neq \uu$ similar argument gives that value of $z$ decreases when we relax $y_l=0$ as well as when we relax $s_l=0$ at
$\xx$. The proof follows.
\end{proof}

Finally, we are ready to prove our main result of this section using Lemmas
\ref{lem:t}, \ref{lem:t1} and \ref{lem:t2}. Together with Lemma \ref{lem:t2},
we will use the fact that on Lemke path $z$ monotonically decreases if $M$ is a
P-matrix or else we get a witness that $M$ is not a
P-matrix~\cite{cottle2009linear}. 

\begin{theorem}
\PLCP reduces to \EOPL in polynomial-time. 
\end{theorem}
\begin{proof}
	Given an instance of $\CI=(\MM,\qq)$ of \PLCP, where $M\in \Real^{d\times d}$ and $\qq\in \Real^{d\times 1}$ reduce it to an instance $\CE$ of \EOPL as described above with vertex set $\vert=\{0,1\}^{2d}$ and procedures $S$, $P$ and $\pot$ as given in Table \ref{tab:S}, \ref{tab:P}, and \ref{tab:F} respectively.

Among solutions of \EOPL instance $\CE$, there is no local potential maximizer,
	i.e., $\uu\neq \vv$ such that $\vv=S(\uu)$, $\uu=P(\vv)$ and $\pot(\uu)>\pot(\vv)$
	due to Lemma \ref{lem:t}. We get a solution $\uu \neq 0$ such that either
	$S(P(\uu))\neq \uu$ or $P(S(\uu))\neq \uu$, then by Lemma \ref{lem:t1} it is
	valid configuration and has a corresponding vertex $\xx=(\yy,\ps,z)$ in
	$\CPol$. Again by Lemma~\ref{lem:t1} if $z=0$ then $\yy$ is a $\PLo$ type solution
	of our \PLCP instance $\CI$. On the other hand, if $z>0$ then from Lemma
	\ref{lem:t2} we get that on both the two adjacent edges to $\xx$ on Lemke
	path the value of $z$ either increases or deceases. This gives us a minor of $M$
	which is non-positive~\cite{cottle2009linear}, 
	i.e., a $\PLt$ type solution of the $\PLCP$	instance $\CI$.
\end{proof}

\newpage


\bibliography{paper}

\newpage
\appendix
\section{\EOML to \EOPL}
\label{sec:EOMLtoEOPL}

Given an instance $\CI$ of $\EOML$ defined by circuits $S,P$ and $V$ on vertex
set $\{0,1\}^n$ we are going to create an instance $\CI'$ of \EOPL with circuits
$S',P'$, and $V'$ on vertex set $\{0,1\}^{(n+1)}$, i.e., we introduce one extra bit.  
This extra bit is essentially to take care of the difference in the value of potential 
at the starting point in \EOML and \EOPL, namely $1$ and $0$ respectively. 

Let $k=n+1$, then we create a potential function $V':\{0,1\}^k \rightarrow
\{0,\dots,2^k-1\}$. 
The idea is to make $0^k$ the starting point with potential zero as required,
and to make all other vertices with first bit $0$ be dummy vertices with self
loops. The real graph
will be embedded in vertices with first bit $1$, i.e., of type $(1,\uu)$. Here
by $(b,\uu)\in \{0,1\}^k$, where $b\in \{0,1\}$ and $\uu\in \{0,1\}^n$, we mean
a $k$ length bit string with first bit set to $b$ and for each $i\in[2:k]$ bit $i$ 
set to bit $u_i$. 

\medskip
\medskip

\noindent{\bf Procedure $V'(b,\uu)$:} If $b=0$ then Return $0$, otherwise Return $V(\uu)$. 
\medskip
\medskip

\noindent{\bf Procedure $S'(b,\uu)$:}
\vspace{-0.3cm}

\begin{enumerate}
\item If $(b,\uu)=0^k$ then Return $(1,0^n)$
\item If $b=0$ and $\uu\neq 0^n$ then Return $(b,\uu)$ (creating self loop for dummy vertices)
\item If $b=1$ and $V(\uu)=0$ then Return $(b,\uu)$ (vertices with zero potentials have self loops)
\item If $b=1$ and $V(\uu)>0$ then Return $(b,S(\uu))$ (the rest follows $S$)
\end{enumerate}

\noindent{\bf Procedure $P'(b,\uu)$:}
\vspace{-0.3cm}

\begin{enumerate}
\item If $(b,\uu)=0^k$ then Return $(b,\uu)$ (initial vertex points to itself in $P'$).
\item If $b=0$ and $\uu\neq 0^n$ then Return $(b,\uu)$ (creating self loop for dummy vertices)
\item If $b=1$ and $\uu=0^n$ then Return $0^k$ (to make $(0,0^n)\rightarrow (1,0^n)$ edge consistent)
\item If $b=1$ and $V(\uu)=0$ then Return $(b,\uu)$ (vertices with zero potentials have self loops)
\item If $b=1$ and $V(\uu)>0$ and $\uu \neq 0^n$ then Return $(b,P(\uu))$ (the rest follows $P$)
\end{enumerate}

Valid solutions of \EOML of type T2 and T3 requires the potential to be strictly greater than zero, while solutions of \EOPL may have zero potential. However, a solution of \EOPL can not be a self loop, so we've added self-loops around vertices with zero potential in the \EOPL instance.
By construction, the next lemma follows:
\begin{lemma}\label{lem:m2p-valid}
$S'$, $P'$, $V'$ are well defined and polynomial in the sizes of $S$, $P$, $V$ respectively. 
\end{lemma}

Our main theorem in this section is a consequence of the following three lemmas.

\begin{lemma}\label{lem:m2p-sl}
For any $\xx=(b,\uu)\in \{0,1\}^k$, $P'(\xx)=S'(\xx)=\xx$ (self loop) iff $\xx\neq 0^k$, and $b=0$ or $V(\uu)=0$.
\end{lemma}
\begin{proof}
This follows by the construction of $V'$, the second condition in $S'$ and $P'$, and third and fourth conditions in $S'$ and $P'$ respectively. 
\end{proof}

\begin{lemma}\label{lem:m2p-r1}
Let $\xx=(b,\uu)\in \{0,1\}^k$ be such that $S'(P'(\xx))\neq \xx \neq 0^k$ or $P'(S'(\xx))\neq \xx$ (an R1 type solution of \EOPL instance $\CI'$), then $\uu$ is a solution of \EOML instance $\CI$.
\end{lemma}
\begin{proof}
The proof requires a careful case analysis. 
By the first conditions in the descriptions of $S',P'$ and $V'$, we have $\xx \neq 0^k$. 
Further, since $\xx$ is not a self loop, Lemma \ref{lem:m2p-sl} implies $b=1$  and $V'(1,\uu)=V(\uu)>0$.
\medskip

\noindent{\em Case I.}
If $S'(P'(\xx))\neq \xx\neq 0^k$ then we will show that either $\uu$ is a genuine start of a line other than $0^n$ giving a T1 type solution of \EOML instance $\CI$, or there is some issue with the potential at $\uu$ giving either a T2 or T3 type solution of $\CI$. Since $S'(P'(1,0^n))=(1,0^n)$, $\uu \neq 0^n$. Thus if $S(P(\uu))\neq \uu$ then we get a T1 type solution of $\CI$ and proof follows. If $V(\uu)=1$ then we get a T2 solution of $\CI$ and proof follows. 

Otherwise, we have $S(P(\uu))=\uu$ and $V(\uu)>1$. Now since also $b=1$ $(1,\uu)$ is not a self loop (Lemma \ref{lem:m2p-sl}). 
Then it must be the case that $P'(1,\uu)=(1,P(\uu))$. However, $S'(1,P(\uu))\neq (1,\uu)$ even though $S(P(\uu))=\uu$. This happens only when $P(\uu)$ is a self loop because of $V(P(\uu))=0$ (third condition of $P'$).
Therefore, we have $V(\uu)-V(P(\uu))>1$ implying that $\uu$ is a T3 type solution of $\CI$. 
\medskip

\noindent{\em Case II.}
Similarly, if $P'(S'(\xx))\neq \xx$, then either $\uu$ is a genuine end of a line of $\CI$, or there is some issue with the potential at $\uu$. If $P(S(\uu))\neq \uu$ then we get T1 solution of $\CI$. Otherwise, $P(S(\uu))=\uu$ and $V(\uu)>0$. Now as $(b,\uu)$ is not a self loop and $V(\uu)>0$, it must be the case that $S'(b,\uu)=(1,S(\uu))$. However, $P'(1, S(\uu))\neq (b,\uu)$ even though $P(S(\uu))=\uu$. This happens only when $S(\uu)$ is a self loop because of $V(S(\uu))=0$. Therefore, we get $V(S(\uu))-V(\uu)<0$, i.e., $\uu$ is a type T3 solution of $\CI$. 
\end{proof}

\begin{lemma}\label{lem:m2p-r2}
Let $\xx=(b,\uu)\in \{0,1\}^k$ be an R2 type solution of the constructed \EOPL instance $\CI'$, then $\uu$ is a type T3 solution of \EOML instance~$\CI$.
\end{lemma}
\begin{proof}
Clearly, $\xx\neq 0^k$. Let $\yy = (b',\uu') = S'(\xx) \neq \xx$, and observe that $P(\yy) = \xx$. This also implies that $\yy$ is not a self loop, and hence $b=b'=1$ and $V(\uu)>0$ (Lemma \ref{lem:m2p-sl}). Further, $\yy = S'(1,\uu)=(1,S(\uu))$, hence $\uu'=S(\uu)$. Also, $V'(\xx)=V'(1,\uu)=V(\uu)$ and $V'(\yy)=V'(1,\uu')=V(\uu')$. 

Since $V'(\yy)-V'(\xx)\le 0$ we get $V(\uu')-V(\uu)\le 0 \Rightarrow V(S(\uu)) - V(\uu) \le 0\Rightarrow V(S(\uu)) - V(\uu)\neq 1$. Given that $V(\uu)>0$, $\uu$ gives a type T3 solution of \EOML.
\end{proof}

\begin{theorem}\label{thm:m2p}
An instance of \EOML can be reduced to an instance of \EOPL in linear time such that a solution of the former can be constructed in a linear time from the solution of the latter. 
\end{theorem}

\section{\EOPL to \EOML}
\label{sec:eopl2eoml}

In this section we give a linear time reduction from an instance $\CI$ of \EOPL to an instance $\CI'$ of \EOML. Let the given \EOPL instance $\CI$ be defined on vertex set $\{0,1\}^n$ and with procedures $S,P$ and $V$, where $V:\{0,1\}^n\rightarrow \{0,\dots,2^m-1\}$. 
\medskip

\noindent{\bf Valid Edge.} We call an edge $\uu \rightarrow \vv$ valid if $\vv=S(\uu)$ and $\uu=P(\vv)$. 
\medskip

We construct an \EOML instance $\CI'$ on $\{0,1\}^k$ vertices where $k=n+m$. 
Let $S',P'$ and $V'$ denotes the procedures for $\CI'$ instance. 
The idea is to capture value $V(\xx)$ of the potential in the $m$ least significant bits of vertex description itself, so that it can be gradually increased or decreased on valid edges. For vertices with irrelevant values of these least $m$ significant bits we will create self loops. Invalid edges will also become self loops, e.g., if $\yy=S(\xx)$ but $P(\yy)\neq \xx$ then set $S'(\xx,.)=(\xx,.)$. We will see how these can not introduce new solutions. 

In order to ensure $V'(0^k)=1$, the $V(S(0^n))=1$ case needs to be discarded. For
this, we first do some initial checks to see if the given instance $\CI$ is not
trivial.  If the input \EOPL instance is trivial, in the sense that either
$0^n$ or $S(0^n)$ is a solution, then we can just return it.

\begin{lemma}
\label{lem:valid-edges}
If $0^n$ or $S(0^n)$ are not solutions of \EOPL instance $\CI$ then $0^n
\rightarrow S(0^n) \rightarrow S(S(0^n))$ are valid edges, and $V(S(S(0^n))\ge 2$. 
\end{lemma}

\begin{proof}
Since both $0^n$ and $S(0^n)$ are not solutions, we have
	$V(0^n)<V(S(0^n))<V(S(S(0^n)))$, $P(S(0^n))=0^n$, and for $\uu = S(0^n)$,
	$S(P(\uu))=\uu$ and $P(S(\uu))=\uu$. In other words, $0^n \rightarrow S(0^n)
	\rightarrow S(S(0^n))$ are valid edges, and since $V(0^n)=0$, we have
	$V(S(S(0^n))\ge 2$. 
\end{proof}

Let us assume now on that $0^n$ and $S(0^n)$ are not solutions of $\CI$, and
then by Lemma \ref{lem:valid-edges}, we have $0^n \rightarrow S(0^n) \rightarrow
S(S(0^n))$ are valid edges, and $V(S(S(0^n))\ge 2$. We can avoid the need to check
whether $V(S(0))$ is one all together, by making $0^n$ point directly to
$S(S(0^n))$ and make $S(0^n)$ a dummy vertex. 

We first construct $S'$ and $P'$, and then construct $V'$ which will give
value zero to all self loops, and use the least significant $m$ bits to give a
value to all other vertices.
Before describing $S'$ and $P'$ formally, we first describe the underlying
principles. Recall that in $\CI$ vertex set is $\{0,1\}^n$ and possible potential values are $\{0,\dots,2^m-1\}$, while in $\CI'$ vertex set is $\{0,1\}^k$ where $k=m+n$. 
We will denote a vertex of $\CI'$ by a tuple $(\uu,\pi)$, where $\uu \in
\{0,1\}^n$ and $\pi\in \{0,\dots,2^m-1\}$. 
Here when we say that we introduce an {\em edge $\xx\rightarrow \yy$} we mean
that we introduce a valid edge from $\xx$ to $\yy$, i.e., $\yy=S'(\xx)$ and $\xx=P(\yy)$. 
\begin{itemize}
\item Vertices of the form $(S(0^n),\pi)$ for any $\pi \in \{0,1\}^m$ and the vertex $(0^n,1)$ are
dummies and hence have self loops.
\item If $V(S(S(0^n))=2$ then we introduce an edge $(0^n,0)\rightarrow(S(S(0^n)),2)$, otherwise 
\begin{itemize}
\item for $p=V(S(S(0^n))$, we introduce the edges $(0^n,0)\ra (0^n,2)\ra (0^n, 3)\dots (0^n,p-1)\ra (S(S(0^n)),p)$.
\end{itemize}
\item If $\uu \ra \uu'$ valid edge in $\CI$ then let $p=V(\uu)$ and $p'=V(\uu')$
\begin{itemize}
\item If $p=p'$ then we introduce the edge $(\uu,p)\ra (\uu',p')$. 
\item If $p<p'$ then we introduce the edges $(\uu,p)\ra (\uu,p+1)\ra \dots\ra (\uu,p'-1)\ra (\uu',p')$.
\item If $p>p'$ then we introduce the edges $(\uu,p)\ra (\uu,p-1)\ra \dots\ra (\uu,p'+1)\ra (\uu',p')$.
\end{itemize}
\item If $\uu\neq 0^n$ is the start of a path, i.e., $S(P(\uu))\neq \uu$, then
make $(\uu,V(\uu))$ start of a path by ensuring $P'(\uu,V(\uu))=(\uu,V(\uu))$.
\item If $\uu$ is the end of a path, i.e., $P(S(\uu))\neq \uu$, then make
$(\uu,V(\uu))$ end of a path by ensuring $S'(\uu,V(\uu))=(\uu,V(\uu))$.
\end{itemize}

Last two bullets above remove singleton solutions from the system by making them
self loops. However, this can not kill all the solutions since there is a path
starting at $0^n$, which has to end somewhere. Further, note that this entire process ensures that no new start or end of a paths are introduced. 
\medskip
\medskip

\noindent{\bf Procedure $S'(\uu,\pi)$.} 
\vspace{-0.2cm}

\begin{enumerate}
\item If ($\uu=0^n$ and $\pi=1$) or $\uu=S(0^n)$ then Return $(\uu,\pi)$. 
\item If $(\uu,\pi)=0^k$, then let $\uu'=S(S(0^n))$ and $p'=V(\uu')$. 
\begin{enumerate}
\item If $p'=2$ then Return $(\uu',2)$ else Return $(0^n,2)$.
\end{enumerate}
\item If $\uu=0^n$ then
\begin{enumerate}
\item If $2\le \pi<p'-1$ then Return $(0^n,\pi+1)$.
\item If $\pi=p'-1$ then Return $(S(S(0^n)),p')$.
\item If $\pi\ge p'$ then Return $(\uu,\pi)$.
\end{enumerate}
\item Let $\uu'=S(\uu)$, $p'=V(\uu')$, and $p=V(\uu)$. 
\item If $P(\uu')\neq \uu$ or $\uu'=\uu$ then Return $(\uu,\pi)$
\item If $\pi=p=p'$ or ($\pi=p$ and $p'=p+1$) or $(\pi=p$ and $p'=p-1$) then Return $(\uu',p')$.
\item If $\pi<p\le p'$ or $p\le p'\le \pi$ or $\pi>p\ge p'$ or $p\ge p'\ge \pi$ then Return $(\uu,\pi)$
\item If $p<p'$, then If $p\le \pi<p'-1$ then Return $(\uu,\pi+1)$. If $\pi=p'-1$ then Return $(\uu',p')$.
\item If $p>p'$, then if $p \ge \pi>p'+1$ then Return $(\uu,\pi-1)$. If $\pi=p'+1$ then Return $(\uu',p')$.
\end{enumerate}
\medskip

\noindent{\bf Procedure $P'(\uu,\pi)$.} 
\vspace{-0.2cm}

\begin{enumerate}
\item If ($\uu=0^n$ and $\pi=1$) or $\uu=S(0^n)$ then Return $(\uu,\pi)$. 
\item If $\uu=0^n$, then 
\begin{enumerate}
\item If $\pi=0$ then Return $0^k$.
\item If $\pi<V(S(S(0^n)))$ and $\pi\notin \{1,2\}$ then Return $(0^n,\pi-1)$.
\item If $\pi<V(S(S(0^n)))$ and $\pi=2$ then Return $0^k$.
\end{enumerate}
\item If $\uu=S(S(0^n))$ and $\pi=V(S(S(0^n))$ then 
\begin{enumerate}
\item If $\pi=2$ then Return $(0^n,0)$, else Return $(0^n,\pi-1)$. 
\end{enumerate}
\item If $\pi=V(\uu)$ then 
\begin{enumerate}
\item Let $\uu'=P(\uu)$, $p'=V(\uu')$, and $p=V(\uu)$. 
\item If $S(\uu')\neq \uu$ or $\uu'=\uu$ then Return $(\uu,\pi)$
\item If $p=p'$ then Return $(\uu',p')$ 
\item If $p'<p$ then Return $(\uu',p-1)$ else Return $(\uu',p+1)$
\end{enumerate}
\item Else \% when $\pi \neq V(\uu)$
\begin{enumerate}
\item Let $\uu'=S(\uu)$, $p'=V(\uu')$, and $p=V(\uu)$
\item If $P(\uu')\neq \uu$ or $\uu'=\uu$ then Return $(\uu,\pi)$
\item If $p'=p$ or $\pi<p< p'$ or $p<p'\le \pi$ or $\pi>p> p'$ or $p>p'\ge \pi$ then Return $(\uu,\pi)$
\item If $p<p'$, then If $p<\pi\le p'-1$ then Return $(\uu,\pi-1)$. 
\item If $p>p'$, then if $p> \pi\ge p'+1$ then Return $(\uu,\pi+1)$. 
\end{enumerate}
\end{enumerate}

As mentioned before, the intuition for the potential function procedure $V'$ is to return zero for self loops, return $1$ for $0^k$, and return the number specified by the lowest $m$ bits for the rest. 
\medskip
\medskip

\noindent{\bf Procedure $V'(\uu,\pi)$.} Let $\xx=(\uu,\pi)$ for notational convenience.
\vspace{-0.2cm}

\begin{enumerate}
\item If $\xx=0^k$, then Return $1$. 
\item If $S'(\xx) = \xx$ and $P'(\xx)=\xx$ then Return $0$.
\item If $S'(\xx) \neq \xx$ or $P'(\xx)\neq \xx$ then Return $\pi$.
\end{enumerate}

The fact that procedures $S'$, $P'$ and $V'$ give a valid \EOML instance follows from construction.
\begin{lemma}\label{lem:p2m-valid}
Procedures $S'$, $P'$ and $V'$ gives a valid \EOML instance on vertex set $\{0,1\}^k$, where $k=m+n$ and $V':\{0,1\}^k\ra \{0,\dots, 2^k-1\}$.
\end{lemma}

The next three lemmas shows how to construct a solution of \EOPL instance $\CI$ from a type T1, T2, or T3 solution of constructed \EOML instance $\CI'$.
The basic idea for next lemma, which handles type T1 solutions, is that we never create spurious end or start of a path. 
\begin{lemma}\label{lem:p2m-t1}
Let $\xx=(\uu,\pi)$ be a type T1 solution of constructed \EOML instance $\CI'$. Then $\uu$ is a type R1 solution of the given \EOPL instance $\CI$.
\end{lemma}

\begin{proof}
Let $\Delta=2^m-1$.
In $\CI'$, clearly $(0^n,\pi)$ for any $\pi \in {1,\dots, \Delta}$ is not a start or end of a path, and $(0^n,0)$ is not an end of a path. Therefore, $\uu\neq 0^n$. Since $(S(0^n),\pi), \forall \pi\in \{0,\dots,\Delta\}$ are self loops, $\uu \neq S(0^n)$.

If to the contrary, $S(P(\uu))=\uu$ and $P(S(\uu))=\uu$. If $S(\uu)=\uu=P(\uu)$ then $(\uu,\pi),\ \forall \pi\in\{0,\dots,\Delta\}$ are self loops, a contradiction. \todo{I don't understand this line. Basically, the point found can't have corresponded to a self-loop because it would then have been a self-loop too.}

For the remaining cases, let $P'(S'(\xx))\neq \xx$, and let $\uu'=S(\uu)$. \todo{There must have been a edge in the original if there is a successor edge in the new instance}. There is a valid edge from $\uu$ to $\uu'$ in $\CI$. Then we will create valid edges from $(\uu,V(\uu))$ to $(S(\uu),V(S(\uu))$ with appropriately changing second coordinates. The rest of $(\uu,.)$ are self loops, a contradiction. 

Similar argument follows for the case when $S'(P'(\xx))\neq \xx$. 
\end{proof}

The basic idea behind the next lemma is that a T2 type solution in $\CI'$ has
potential $1$. Therefore, it is surely not a self loop. Then it is either an end of a path or near an end of a path, or else near a potential violation. 

\begin{lemma}\label{lem:p2m-t2}
Let $\xx=(\uu,\pi)$ be a type T2 solution of $\CI'$. Either $\uu \neq 0^n$ is start of a path in $\CI$ (type R1 solution), or $P(\uu)$ is an R1 or R2 type solution in $\CI$, or $P(P(\uu))$ is an R2 type solution in $\CI$.
\end{lemma}

\begin{proof}
Clearly $\uu \neq 0^n$, and $\xx$ is not a self loop, i.e., it is not a dummy vertex with irrelevant value of $\pi$. Further, $\pi=1$. If $\uu$ is a start or end of a path in $\CI$ then done. 

Otherwise, if $V(P(\uu))>\pi$ then we have $V(\uu)\le \pi$ and hence $V(\uu)-V(P(\uu))\le 0$ giving $P(\uu)$ as an R2 type solution of $\CI$. 
If $V(P(\uu))<\pi=1$ then $V(P(\uu))=0$. Since potential can not go below zero, either $P(\uu)$ is an end of a path, or for $\uu''=P(P(\uu))$ and $\uu'=P(\uu)$ we have $\uu'=S(\uu'')$ and $V(\uu')-V(\uu'')\le 0$, giving $\uu''$ as a type R2 solution of $\CI$.
\end{proof}

At a type T3 solution of $\CI'$ potential is strictly positive, hence these solutions are not self loops. If they correspond to potential violation in $\CI$ then we get a type R2 solution. But this may not be the case, if we made $S'$ or $P'$ self pointing due to end or start of a path respectively. In that case, we get a type R1 solution. The next lemma formalizes this intuition. 

\begin{lemma}\label{lem:p2m-t3}
Let $\xx=(\uu,\pi)$ be a type T3 solution of $\CI'$. If $\xx$ is a start or end of a path in $\CI'$ then $\uu$ gives a type R1 solution in $\CI$. Otherwise $\uu$ gives a type R2 solution of $\CI$.
\end{lemma}

\begin{proof}
Since $V'(\xx)>0$, it is not a self loop and hence is not dummy, and $\uu\neq 0^n$. If $\uu$ is start or end of a path then $\uu$ is a type R1 solution of $\CI$. Otherwise, there are valid incoming and outgoing edges at $\uu$, therefore so at $\xx$. 

If $V((S(\xx))-V(\xx)\neq 1$, then since potential either remains the same or increases or decreases exactly by one on edges of $\CI'$, it must be the case that $V(S(\xx))-V(\xx)\le 0$. This is possible only when $V(S(\uu))\le V(\uu)$. Since $\uu$ is not an end of a path we do have $S(\uu)\neq \uu$ and $P(S(\uu))=\uu$. Thus, $\uu$ is a type T2 solution of $\CI$.

If $V((\xx)-V(P(\xx))\neq 1$, then by the same argument we get that for $(\uu'',\pi'')=P(\uu)$, $\uu''$ is a type R2 solution of $\CI$. 
\end{proof}

Our main theorem follows using Lemmas \ref{lem:p2m-valid}, \ref{lem:p2m-t1}, \ref{lem:p2m-t2}, and \ref{lem:p2m-t3}.

\begin{theorem}\label{thm:p2m}
An instance of \EOPL can be reduced to an instance of \EOML in polynomial time such that a solution of the former can be constructed in a linear time from the solution of the latter. 
\end{theorem}

\section{Pseudo-code for Lemke's algorithm}
\label{app:lemke}

\begin{tabular}{|l|}
\hline
\hspace{5pt} {\bf If} $\qq\ge 0$ {\bf then} {\bf Return} $\yy\leftarrow \zeros$ \\
\hspace{5pt} $\yy\leftarrow 0, z\leftarrow |\min_{i \in [d]} q_i|, \ps=\qq+z\ones$\\
\hspace{5pt} $i\leftarrow $ duplicate label at vertex $(\yy,\ps,z)$ in $\CPol$. $flag\leftarrow 1$ \\
\hspace{5pt} {\bf While} $z>0$ {\bf do}\\
\hspace{10pt} {\bf If} $flag=1$ {\bf then} set $(\yy',\ps',z')\leftarrow $ vertex obtained by relaxing $y_i=0$ at $(\yy,\ps,z)$ in $\CPol$\\
\hspace{10pt} {\bf Else} set $(\yy',\ps',z')\leftarrow $ vertex obtained by relaxing $s_i=0$ at $(\yy,\ps,z)$ in $\CPol$\\
\hspace{10pt} {\bf If} $z>0$ {\bf then}\\
\hspace{15pt} $i \leftarrow $ duplicate label at $(\yy',\ps',z')$\\
\hspace{15pt} {\bf If} $v_i>0$ and $v'_i=0$ {\bf then} $flag\leftarrow 1$. {\bf Else} $flag\leftarrow 0$\\
\hspace{15pt} $(\yy,\ps,z)\leftarrow(\yy',\ps',z')$\\
\hspace{5pt} End {\bf While} \\
\hspace{5pt} {\bf Return} $\yy$\\
\hline
\end{tabular}

\newpage

\section{Missing Procedures and Proofs from Section~\ref{sec:PLCPtoEOPL}}\label{app:PLCPtoEOPL}
\subsection{Procedures $\isvalid$, $\ite$, and $\eti$}\label{app:proc}
\begin{table}[!hbt]
\caption{Procedure \isvalid(\uu)}\label{tab:iv}
\begin{tabular}{|l|}
\hline
\hspace{5pt} {\bf If} $\uu=0^{n}$ {\bf then} {\bf Return} 1\\
\hspace{5pt} {\bf Else} let $\tau = (u_{(d+1)}+\dots+u_{2d})$\\
\hspace{15pt} {\bf If} $\tau> 1$ {\bf then} {\bf Return} 0\\
\hspace{15pt} Let $S\leftarrow \emptyset$. \% set of tight inequalities. \\
\hspace{15pt} {\bf If} $\tau = 0$ {\bf then} $S=S\cup \{ z=0\}$. \\
\hspace{15pt} {\bf Else}\\
\hspace{30pt} Set $l\leftarrow $ index of the non-zero coordinate in vector $(u_{(d+1)},\dots,u_{2d})$. \\
\hspace{30pt} Set $S=\{y_l=0, s_l=0\}$.\\
\hspace{15pt} {\bf For} each $i$ from $1$ to $d$ {\bf do} \\
\hspace{30pt} {\bf If} $u_i=0$ {\bf then} $S=S\cup \{y_i=0\}$, {\bf Else} $S=S\cup \{s_i=0\}$\\
\hspace{15pt} Let $A$ be a matrix formed by lhs of equalities $M\yy+\ps -\ones z=\qq$ and that of set $S$\\
\hspace{15pt} Let $\bb$ be the corresponding rhs, namely $\bb=[\qq; \zeros_{d\times 1}]$.\\
\hspace{15pt} Let $(\yy',\ps',z') \leftarrow \bb * A^{-1}$\\
\hspace{15pt} {\bf If} $(\yy',\ps',z') \in \CPol$ {\bf then} {\bf Return} 1, {\bf Else} {\bf Return} 0 \\
\hline
\end{tabular}
\end{table}

\begin{table}[!hbt]
\caption{Procedures \ite(\uu) and \eti(\yy,\ps,z)}\label{tab:ei}
\begin{tabular}{|l|}
\hline
\begin{tabular}{l}
$\ite(\yy,\ps,z)$ \\ \hline
\hspace{5pt} {\bf If} $\exists i \in [d]$ s.t. $y_i * s_i \neq 0$ {\bf then} {\bf Return} $(\zeros_{(2d-2)\times 1};1;1)$ \% Invalid \\
\hspace{5pt} Set $\uu\leftarrow \zeros_{2d\times 1}$. Let $DL=\{i\in [d]\ |\ y_i=0\mbox{ and } s_i=0\}$.\\
\hspace{5pt} {\bf If} $|DL|>1$ {\bf then} {\bf Return} $(\zeros_{(2d-2)\times 1};1;1)$ \%In valid \\
\hspace{5pt} {\bf If} $|DL|=1$ {\bf then} for $i\in DL$, set $u_i\leftarrow 1$\\
\hspace{5pt} {\bf For} each $i\in [d]$ {\bf If} $s_i=0$ {\bf then} set $u_{d+i}\leftarrow 1$\\
\hspace{5pt} {\bf Return} $\uu$
\end{tabular}
\\ \hline
\begin{tabular}{l}
$\eti(\uu)$  \\ \hline
\hspace{5pt} {\bf If} $\uu=0^n$ {\bf then} {\bf Return} $(\zeros_{d \times 1}, \qq+z^0+1, z^0+1)$ \% This case will never happen\\
\hspace{5pt} {\bf If} \isvalid(\uu)=0 {\bf then} {\bf Return} $\zeros_{(2d+1) \times 1}$\\
\hspace{5pt} Let $\tau = (u_{(d+1)}+\dots+u_{2d})$\\
\hspace{5pt} Let $S\leftarrow \emptyset$. \% set of tight inequalities. \\
\hspace{5pt} {\bf If} $\tau = 0$ {\bf then} $S=S\cup \{ z=0\}$. \\
\hspace{5pt} {\bf Else}\\
\hspace{15pt} Set $l\leftarrow $ index of non-zero coordinate in vector $(u_{(d+1)},\dots,u_{2d})$. \\
\hspace{15pt} Set $S=\{y_l=0, s_l=0\}$.\\
\hspace{5pt} {\bf For} each $i$ from $1$ to $d$ {\bf do} \\
\hspace{15pt} {\bf If} $u_i=0$ {\bf then} $S=S\cup \{y_i=0\}$, {\bf Else} $S=S\cup \{s_i=0\}$\\
\hspace{5pt} Let $A$ be a matrix formed by lhs of equalities $M\yy+\ps -\ones z=\qq$ and that of set $S$\\
\hspace{5pt} Let $\bb$ be the corresponding rhs, namely $\bb=[\qq; \zeros_{d\times 1}]$.\\
\hspace{5pt} {\bf Return} $\bb * A^{-1}$\\
\end{tabular}\\
\hline
\end{tabular}
\end{table}

\subsection{Proof of Lemma~\ref{lem:vert}}

\textbf{Lemma \ref{lem:vert}} (restated) : \emph{
If $\isvalid(\uu)=1$ then $\uu=\ite(\eti(\uu))$, and the corresponding vertex $(\yy,\ps,z)\in \eti(\uu)$ of $\CPol$ is feasible in (\ref{eq:c}). If $(\yy,\ps,z)$ is a feasible vertex of (\ref{eq:c}) then $\uu=\ite(\yy,\ps,z)$ is a valid configuration, {\em i.e.,} $\isvalid(\uu)=1$.
}
\begin{proof}
The only thing that can go wrong is that the matrix $A$ generated in $\isvalid$ and $\eti$ procedures are singular, or the set of double labels $DL$ generated in $\ite$ has more than one elements. 
Each of these are possible only when more than $2d+1$ equalities of $\CPol$ hold at the corresponding point $(\yy,\ps,z)$, violating non-degeneracy assumption. 
\end{proof}

\subsection{Proof of Lemma~\ref{lem:PSF}}

\textbf{Lemma \ref{lem:PSF}} (restated) : \emph{
Functions $P$, $S$ and $\pot$ of instance $\CE$ are well defined, making $\CE$ a valid \EOPL instance. 
}
\begin{proof}
Since all three procedures are polynomial-time in $\CL$, they can be defined
by $poly(\CL)$-sized Boolean circuits. Furthermore, for any $\uu \in \vert$,
we have that $S(\uu),P(\uu) \in \vert$. For~$\pot$, 
since the value of $z \in [0,\ \Delta-1]$, we
have $0\le \Delta^2(\Delta-z)\le \Delta^3$. Therefore, $\pot(\uu)$ is an
integer that is at most $2 \cdot \Delta^3$ and hence is in set $\{0,\dots, 2^m-1\}$. 
\end{proof}

\subsection{Proof of Lemma~\ref{lem:pot}}

\textbf{Lemma \ref{lem:pot}} (restated) : \emph{
Let $\uu \neq \uu'$ be two valid configurations, i.e.,
	$\isvalid(\uu)=\isvalid(\uu')=1$, and let $(\yy,\ps,z)$ and $(\yy',\ps',z')$
	be the corresponding vertices in $\CPol$. Then the following holds: $(i)$
	$\pot(\uu)=\pot(\uu')$ iff $z=z'$. $(ii)$ $\pot(\uu)>\pot(\uu')$ iff $z<z'$.
}
\begin{proof}
Among the valid configurations all except $\zeros$ has positive $\pot$ value. Therefore, wlog let $\uu,\uu'\neq \zeros$. For these we have $\pot(\uu)=\lfloor \Delta^2*(\Delta -z)\rfloor$, and $\pot(\uu')=\lfloor \Delta^2*(\Delta -z')\rfloor$. 

Note that since both $z$ and $z'$ are coordinates of vertices of $\CPol$, whose description has highest coefficient of $\max\{\max_{i,j\in [d]} M(i,j),\max_{i\in [d]} |q_i|\}$, and therefore their numerator and denominator both are bounded above by $\Delta$. Therefore, if $z< z'$ then we have 
\[
z'-z\ge \frac{1}{\Delta^2} \Rightarrow ((\Delta-z) - (\Delta - z'))*\Delta^2 \ge 1 \Rightarrow \pot(\uu)-\pot(\uu') \ge 1.
\]

For $(i)$, if $z=z'$ then clearly $\pot(\uu)=\pot(\uu')$, and from the above argument it also follows that if $\pot(\uu)= \pot(\uu')$ then it can not be the case that $z\neq z'$. Similarly for $(ii)$, if $\pot(\uu)>\pot(\uu')$ then clearly, $z'>z$, and from the above argument it follows that if $z'>z$ then it can not be the case that $\pot(\uu')\ge \pot(\uu)$. 
\end{proof}

\end{document}